\documentclass[conference]{IEEEtran}
\IEEEoverridecommandlockouts
\IEEEpubid{\begin{minipage}{\textwidth}\ \\[12pt]
\textcopyright~2026 IEEE. Personal use of this material is permitted. Permission from IEEE must be obtained for all other uses, in any current or future media, including reprinting/republishing this material for advertising or promotional purposes, creating new collective works, for resale or redistribution to servers or lists, or reuse of any copyrighted component of this work in other works.
\end{minipage}}
\usepackage{cite}
\usepackage{amsmath,amssymb,amsfonts}
\usepackage{comment}
\usepackage{algorithm}
\usepackage{algpseudocode}
\usepackage{graphicx}
\usepackage{textcomp}
\usepackage{xcolor}
\usepackage{amsthm}
\usepackage{url}
\usepackage{multirow}
\usepackage{booktabs}
\usepackage{ipdps26repro}

\newtheorem{lemma}{Lemma}
\newtheorem{definition}{Definition}

\def\BibTeX{{\rm B\kern-.05em{\sc i\kern-.025em b}\kern-.08em
    T\kern-.1667em\lower.7ex\hbox{E}\kern-.125emX}}
\begin{document}

\title{SMART-MIG: A Learning Framework for Scalable and Energy-Efficient GPU Scheduling\\

\author{\IEEEauthorblockN{Wenqing Yu*\thanks{*The first two authors contributed equally}, Neel Karia*, Tanvi Hisaria, Clifford Stein}
\IEEEauthorblockA{\textit{Columbia University, USA} \\
\{wy2462,nmk2154,th2720,cs2035\}@columbia.edu}
\and
\IEEEauthorblockN{Olivier Tardieu, Asser Tantawi}
\IEEEauthorblockA{\textit{IBM TJ Watson Research Center, USA} \\
\{tardieu,tantawi\}@us.ibm.com}
}

}
\maketitle

\begin{abstract}
The emergence of Multi-Instance GPU (MIG) technology enables us to run smaller machine learning models on partitions of a GPU rather than the entire device, thus improving utilization and reducing energy consumption, albeit with potential performance trade-offs. Meanwhile, the growing energy demands of GPU-equipped data centers motivate the development of online partitioning and scheduling schemes that not only ensure fast job processing but also achieve high energy efficiency. However, achieving energy–tardiness efficiency with manageable algorithmic complexity in large-scale scheduling remains a great challenge, due to the dual objectives of deciding on the GPU partitions and scheduling jobs onto the slices of the heterogeneous partitions. To address this challenge, we propose SMART-MIG, a parallel computing system that combines Mean-Field Multi-Agent Reinforcement Learning (MF-MARL) for large-scale MIG repartitioning with tailored heuristic algorithms for job scheduling. We demonstrate that the complexity of the repartitioning component remains constant even as the number of jobs and GPUs increases. We also establish theoretical lower bounds on energy consumption and tardiness to rigorously benchmark system performance. Finally, extensive experiments show that SMART-MIG improves the energy–tardiness efficiency by $18\%$ compared to its corresponding static partitioning counterpart, while being only $27\%$ above the theoretical lower bound on energy consumption.
\end{abstract}

\begin{IEEEkeywords}
Scheduling, Energy-Efficiency, Repartitioning, Tardiness, GPU, MIG, PPO, Mean Field, Multi-Agent Reinforcement Learning, A100, Algorithms, Lower Bounds
\end{IEEEkeywords}

\section{Introduction}

    In recent years, the increasing adoption of artificial intelligence (AI) and machine learning (ML), especially driven by large language models (which require large-scale parallel processing), 
    has advanced the deployment of GPU clusters. These clusters form an important component of modern data centers, and they often run longer-running training jobs and latency-sensitive inference jobs in a co-located manner to increase utilization and improve throughput. According to \cite{OffuttZhu2025DataCenters}, data centers accounted for $4.4\%$ of total electricity consumption in the United States in 2023 (excluding cryptocurrency), and this number is expected to increase to $12\%$ in 2028. About $40\%$ of this energy was used by computing systems such as CPUs and GPUs, and another $40\%$ was used for cooling. The tremendous scale of this issue creates a pressing need for a GPU scheduling framework that balances energy efficiency with SLO compliance (there is often a trade-off) 
    to create a sustainable, high-performance AI infrastructure for the future.
    \IEEEpubidadjcol
    There are several reasons for the high energy consumption of AI workloads on GPUs, including (i) underutilization of GPUs, (ii) fragmented or poor job allocations, (iii) fluctuations in demand, and (iv) the requirements of complex cooling systems. 
    In this work, some of our scheduling algorithms are designed to address the first two issues by intelligently assigning jobs to slices, and our RL-based repartitioning scheme helps resolve the third issue by learning patterns in arrival rates.

    For multi-tenant environments, NVIDIA recently introduced Multi-Instance GPU (MIG) technology across several of its GPUs, including the A100 and the H100. MIG technology allows a GPU to be partitioned into up to seven isolated instances or slices, each with its own compute, cache, and memory resources. A scheduler for MIG has to make two intertwined sets of decisions: (i) which job to allocate to which slice of which MIG partition, and (ii) when and how to repartition a GPU to be more suitable for the job mix. Repartitioning adds a temporal dimension that can reduce energy at the cost of possible job disruption. 

    Previous work on GPU scheduling has emphasized throughput, fairness, utilization, and SLOs \cite{ye2024deep}, while energy studies rely largely on DVFS or cluster-level schedulers that treat GPUs as uniform units. Such approaches overlook MIG, which enables dynamic resource splitting, which is particularly useful for today’s smaller models like Phi-4 \cite{abdin2024phi} and Gemma-2 \cite{team2024gemma}. Recent studies \cite{li2022miso, tan2021serving, zhang2024miger, tang2025pcie} explore MIG-based placement and repartitioning, demonstrating the feasibility of runtime adaptation, but primarily target utilization and performance rather than jointly addressing energy and performance trade-offs.
    
    The repartitioning of MIGs is highly complex as it must consider online decision-making, the nonlinear throughputs of workloads, and the trade-offs between repartitioning overhead and potential performance gains, all while managing energy efficiency. Furthermore, the nature of data centers necessitates scheduling at a massive scale -- a large number of GPUs must handle extremely diverse jobs characterized by heterogeneity, high uncertainty, and large problem sizes. The challenge is to deliver high-quality solutions to this NP-hard problem (since restricted versions of it are also NP-hard \cite{sitters2005complexity}) in seconds.

    Previous studies \cite{10.1145/3688351.3689156, you2023zeus} have relied mainly on simple heuristics or limited the number of reconfigurable MIGs to allow exhaustive search or basic linear programming (LP) formulations. However, these approaches are inadequate in highly uncertain environments, as they cannot effectively capture workload characteristics, adapt to diverse configurations, or scale to large jobs and GPU pools without severe performance degradation and unacceptable computation time. A recent paper has explored reinforcement learning (RL) to address online scheduling, using Deep Q-Networks to learn task characteristics and maintain service quality \cite{lipe2025energy}. However, because the computational complexity of traditional RL scales exponentially with the number of GPUs and jobs, these approaches can only handle very small-scale scenarios (such as scheduling on a single GPU or a few GPUs).
    
    Motivated by these challenges, we pose a fundamental research question: \textbf{Is it possible to design an efficient, multi-objective, and scalable MIG-based scheduling framework for industrial-scale computing that can adapt to uncertainty and real-time demands, while producing high-quality global optimization solutions in seconds?}

    We address this challenge by introducing \textbf{SMART-MIG:} \textbf{S}cheduling with \textbf{M}ean-field Multi-\textbf{A}gent \textbf{R}L using \textbf{T}op-k sampling for \textbf{MIG}, an intelligent real-time scheduling system. The SMART-MIG pipeline is shown in Figure \ref{fig:pipeline}. It builds on the key observation that all MIGs are interchangeable; the scheduling process only needs to determine how many MIGs are assigned to each valid configuration, rather than tracking the details of each MIG. Accordingly, the input to the central controller is represented as a distribution of GPU and job states, while the output is a randomized policy that can be uniformly applied to all MIGs. 
    
    This formulation bounds the size of the input–output space, making it much less dependent on the number of GPUs and jobs. Within this tractable space, the strong representational power of neural networks enables the learning of highly efficient, generalizable RL policies for repartitioning MIGs, regardless of system scale. Once the configuration of each MIG is decided, since the subsequent multi-machine job scheduling problem has a complex action space, we employ a carefully designed efficient heuristic algorithm that achieves competitive performance.

    To rigorously evaluate system performance, we investigate a theoretical lower bound for the MIG scheduling problem in terms of energy and tardiness, where tardiness is defined as the maximum of $0$ and the difference between the job completion time and the deadline. (In other words, jobs that end before  their deadlines have $0$ tardiness, while jobs that are late by $x$ units of time have a tardiness of $x$.)  We then compare algorithmic results with their corresponding bounds, substantially reducing evaluation biases arising from dataset differences. This contribution establishes \textbf{a standardized benchmark} for fair and consistent evaluation.

    In summary, our contributions are as follows.
    \begin{enumerate}
        \item We observe how energy and tardiness vary with the number of GPUs and note diminishing returns as the number of GPUs increases. This fact can help estimate the number of GPUs needed (Subsection \ref{sec:tradeoff}). 
        \item We propose \textbf{SMART-MIG}, a large-scale MIG scheduling framework that integrates methodologies from ML and OR. The framework consists of two parts:
        \begin{enumerate}
            \item We design EDF-based scheduling algorithms to assign jobs to MIG slices across multiple GPUs (without repartitioning), leveraging job throughput characteristics and MIG’s concave power curve (Subsection \ref{sec:algo}).
            \item We employ Mean-Field Multi-Agent Reinforcement Learning (MF-MARL), suitable for large-scale problems, together with Top-k sampling to repartition GPUs (Subsection \ref{sec:controller}).
        \end{enumerate}
        \item We develop lower bounds for energy and tardiness that serve as references for evaluating our scheduling policies (Section \ref{sec:lower_bounds}).
        \item We conduct extensive experiments against theoretical lower bounds to validate the effectiveness of our method. Under high job arrival rates, the traditional no-partition GPU model rapidly deteriorates in performance. In contrast, our static repartitioning scheduler achieves robust results, operating within $\sim 30\%$ of the lower bounds of the energy and tardiness on average. In addition, incorporating dynamic repartitioning based on MF-MARL yields an additional $18\%$ improvement in energy–tardiness efficiency. Importantly, our method demonstrates consistent advantages across diverse workloads: even under low arrival rates, it reduces tardiness by $47\%$ compared to the no-MIG baseline (Section \ref{sec:exp}).
    \end{enumerate}
    \begin{figure}[ht] 
    \centering 
    \includegraphics[width=\linewidth]{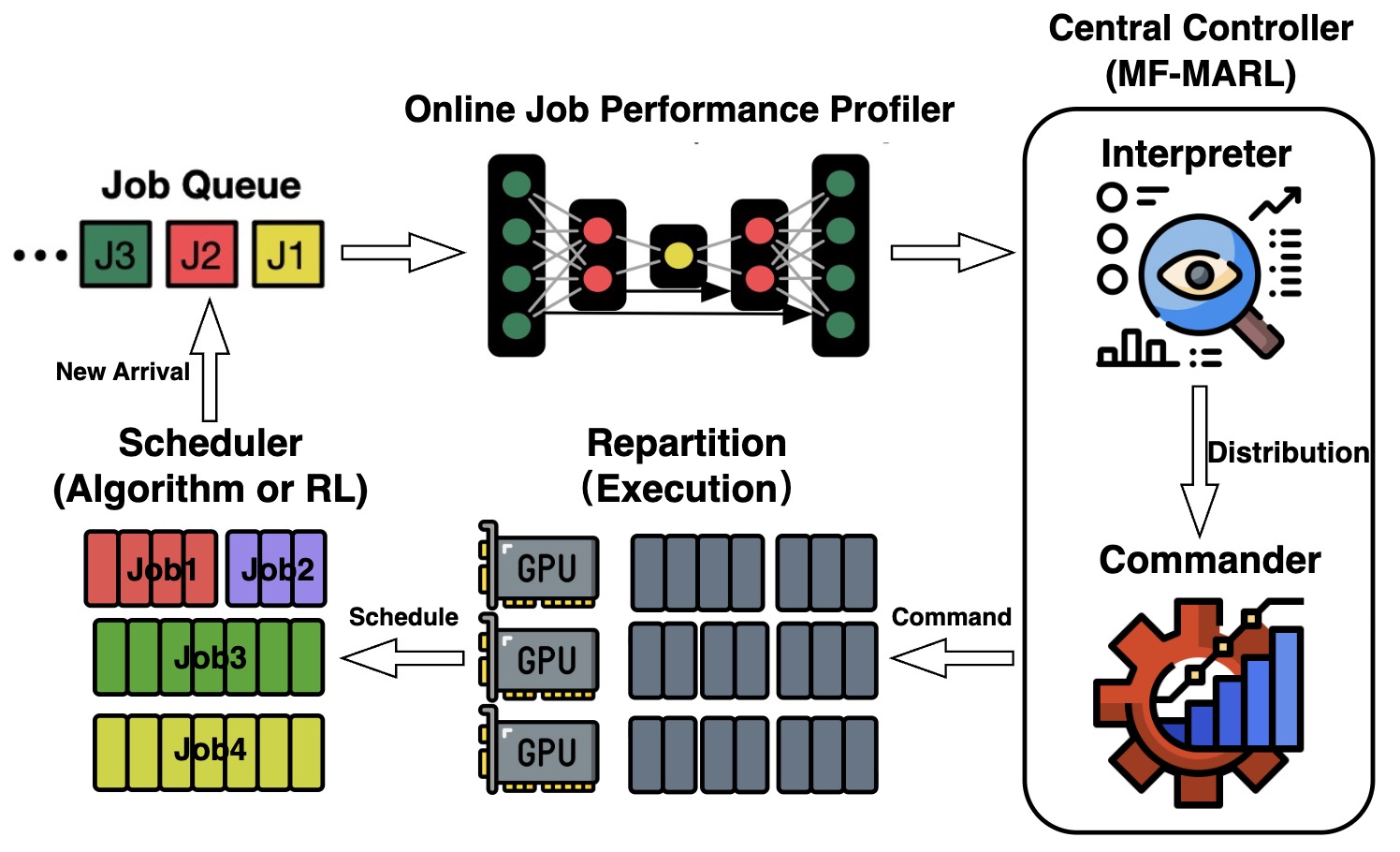} 
    \caption{Method Pipeline Overview. For each task, we extract the job profile, then deploy the MF-MARL model at the central controller to read the state distributions of jobs and MIGs. The controller issues repartitioning instructions, after which a scheduling algorithm assigns jobs to different slices.} 
    \label{fig:pipeline} 
\end{figure}

\section{Related Work}
In this section, we further elaborate on the related work by categorizing it and noting our added contributions.

    \subsection{Energy Aware GPU Scheduling for AI/ML Workloads}
    GPU scheduling has been widely studied in both single-node and cluster contexts. Previously, Gandiva \cite{xiao2018gandiva} introduced cluster scheduling for heterogeneous workloads with time-sharing and job packing to improve utilization by making use of job migration, while recent surveys such as \cite{ye2024deep} highlight scheduling approaches for AI training and inference that focus on meeting performance targets (SLOs) and efficient job placement in modern GPU data centers. 
    DVFS \cite{tang2019impact} and PowerFlow \cite{gu2023energy} reduce energy consumption by integrating energy budgets into scheduling, enabling faster job completion within power limits. However, most of these methods treat GPUs as static resources and overlook dynamic partitioning capabilities like NVIDIA’s MIG. Our work builds on these ideas by addressing tardiness while explicitly prioritizing energy efficiency in dynamically reconfigurable GPU environments.
    
    \subsection{MIG-Aware Scheduling and Repartitioning}
    There has been extensive recent research on MIG-based scheduling. MISO \cite{li2022miso}, for example, makes use of Multi-Process Service (MPS) to decide the best MIG configurations by profiling the job throughputs before making costly switches. MIG-Serving \cite{tan2021serving} looks at deep neural network serving as a scheduling challenge where GPUs can be reshaped on the fly with a two-phase algorithm. 
    MIGER \cite{zhang2024miger} combines MIG and MPS within a GPU to better handle a mixture of online and offline jobs. Other studies \cite{tang2025pcie} explore scheduling while accounting for PCIe bandwidth limits or allowing dynamic layout changes, while mapping out MIG’s valid setups and practical repartitioning constraints. 
    
    \cite{villarrubia2025leveraging} introduces a 3-phase approach called FAR that aims to optimize for makespan, whereas 
    \cite{turkkan2024optimal} provides a method to first assign jobs to slices and then rearrange them as needed (in addition to repartitioning) to minimize the number of GPUs. 
    \cite{lee2024parvagpu} also optimizes to minimize the number of GPUs while maintaining the SLO requirements. 

    \subsection{RL for Resource Management}
    RL has often been used in resource management, starting with \cite{mao2016resource}. Mean Field Multi Agent RL gives tractable approximations for large populations of agents (here, GPUs) \cite{yang2018mean}, and has been used for mean field games \cite{subramanian2019reinforcement}, while Proximal Policy Optimization (PPO) provides stable on-policy learning \cite{schulman2017proximal}. We use them together to obtain scalable, feedback-driven repartitioning policies in a collection of MIGs. 

    \subsection{Lower Bounds in Scheduling}
    Lower bounds help us benchmark the performance of our scheduling policies. There are several classical results for bound flow time and tardiness under various scheduling policies \cite{pinedo2012scheduling}. We derive new lower bounds for energy and tardiness, where the former reduces to makespan minimization, and the latter employs a time-indexed mixed-integer program (MIP).

\section{Background}
In this section, we briefly explain the existing technology and techniques used in this paper.
\subsection{Multi-Instance GPUs}
We first introduce Multi-Instance GPU, a form of resource sharing in GPUs as an alternative to Multi-Process Service (MPS). We focus on the A100-40 GB GPUs, but our techniques can be extended to other similar MIG-enabled GPUs.

NVIDIA’s Multi-Instance GPU (MIG), available on the A100-40GB GPU, partitions the device into up to seven independent isolated slices. A MIG-enabled GPU supports several partitioning configurations by fusing slices into groups of sizes 1, 2, 3, 4, or 7. Each slice is allocated its dedicated hardware resources, including streaming multiprocessors, memory, and cache, ensuring predictable performance and strong fault isolation. This makes MIG essential for multi-tenant cloud environments and workloads that require guaranteed Quality of Service (QoS), as a fault in one instance does not impact others. In this work, we focus on configurations that fully utilize the 40GB of memory on the A100-40GB (Figure \ref{fig:mig}), while ignoring redundant permutations of identical slices in different orders.

\begin{figure}[ht] 
    \centering 
    \includegraphics[width=\linewidth]{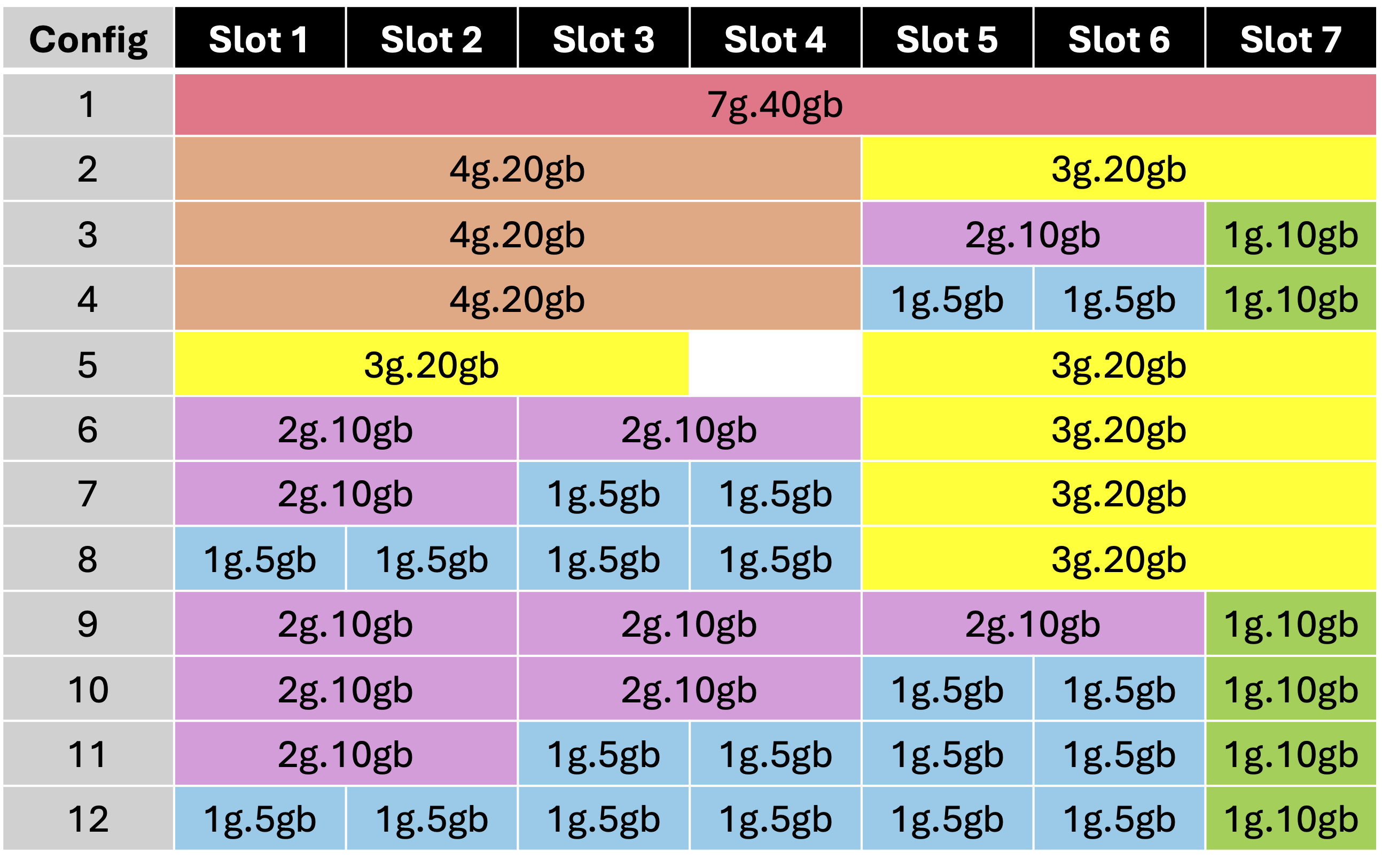} 
    \caption{Nontrivial Configurations of A100-40GB MIG} 
    \label{fig:mig} 
\end{figure}

Complementing MIG is the Multi-Process Service (MPS), a software technology that operates at a different level. While MIG constructs safe hardware slices for different users or large tasks, MPS can be enabled within a single MIG slice or on the entire GPU. It allows multiple CUDA processes from a single user to share that instance's resources concurrently, reducing kernel launch overhead and maximizing utilization.

\subsection{Proximal Policy Optimization}
In the actor-critic framework, the critic estimates the value function to guide updates, while the actor learns the policy that maps states to actions. Based on this framework, the policy-gradient RL algorithm Proximal Policy Optimization (PPO) optimizes a clipped surrogate objective function to iteratively improve the agent’s policy. This approach balances exploration and exploitation while avoiding excessively large policy updates that could destabilize training.

\subsection{Multi-Agent Reinforcement Learning}

We consider a cooperative MARL setting with $N$ agents and a discount factor $\gamma$. At each time step $t$, every agent observes the global state $\boldsymbol{s}_t\in S$, where $S=S_1\times\cdots\times S_N\times S_{sys}$ and $S_i$ denotes the local state of agent $i$, and $S_{sys}$ is the system state. Each agent then selects an action from its local action space $A_i$, which yields the joint action $\boldsymbol{a}_t = (a_t^1,\cdots, a_t^N)\in A=A_1\times\cdots\times A_N$. The system transitions according to a stochastic kernel, and all agents share a joint reward $r_t^i(\boldsymbol{s}_t,\boldsymbol{a}_t)$. The objective is to find a Pareto-optimal joint policy that maximizes the long-term expected reward:

\begin{equation}\label{eq:exp_rew}
\boldsymbol{\pi}^{\star} = \arg\max_{\boldsymbol{\pi}}\mathbb{E}\!\left[ \sum_{t=0}^{\infty} \gamma^{t} r_{t}(\boldsymbol{s}_t,\boldsymbol{a}_t) \right].
\end{equation}

The joint state and action spaces in MARL grow exponentially with the number of agents $N$,   because:
\begin{equation}\label{eq:states}
|S| \times |A| = \left( \prod_{i=1}^{N} |S_{i}| \right) \times \left( \prod_{i=1}^{N} |A_{i}| \right)\times|S_{sys}|.
\end{equation}

This is the well-known curse of dimensionality. Such a rapidly growing state-action space makes data sampling extremely difficult (there can be millions of states and thousands of agents in practice), and a huge amount of data is required for policy evaluation, which makes these algorithms extremely difficult to scale. In practical data center scheduling, where large numbers of MIGs and jobs must be managed, this challenge motivates us to adopt advances from operations research, specifically Mean Field MARL \cite{yang2018mean}.

\subsection{Mean Field MARL}\label{sec:mf-marl}
If all agents are identical, indistinguishable, and interchangeable, that is, we care only about whether an appropriate action is taken, rather than which agent takes it, then we can adopt the Cooperative Mean Field MARL (MF-MARL) framework. In this framework, all agents share the same state space $S$ and action space $A$. The anonymity of agents allows us to represent their state and action information as distributions, i.e., the number of agents occupying each state. Each agent's decision depends on the other agents only through the empirical distribution of their state–action pairs. Let $\mathcal{P}(S)$ and $\mathcal{P}(A)$ be the probability measure spaces over the state space $S$ and the action space $A$, respectively. The empirical distribution of the states is $\mu_{t}(s) = \frac{\sum_{j=1}^{N} \mathbf{1}\!\left( s_{t}^{j} = s \right)}{N} \in \mathcal{P}(S)
$ and the empirical distribution of the actions is $\nu_{t}(a) = \frac{\sum_{j=1}^{N} \mathbf{1}\!\left( a_{t}^{j} = a \right)}{N} \in \mathcal{P}(A)
$.

Formally, at each timestep $t$, the central controller observes the state distribution $\mu_t$ and outputs a policy $\pi_t(\cdot,\mu_t)$, which is a mapping from $S$ to $\mathcal{P}(A)$. Since all agents are interchangeable, we consider the representative agent who selects an action $a_t\sim \pi_t(s_t,\mu_t)$ based on its local state $s_t$, receives a reward $r_t(s_t,\mu_t,a_t,\nu_t)$, and transitions to the next state $s_{t+1}\sim P_t(s_t,\mu_t,a_t,\nu_t)$. Here, both the transition probability $P_t$ and the reward $r_t$ depend on the state and action distributions $\nu_{t}(\cdot) := \int_{\mathcal{S}} \pi_{t}(s, \mu_{t})(\cdot)\,\mu_{t}(s)\, ds$. 

For notational simplicity, let $h_t(\cdot)=\pi_t(\cdot,\mu_t)$, and let $\mathcal{H}:=\{h:S\to \mathcal{P}(A)\}$ be the function space of $h_t(\cdot)$. The policy learned by the central controller thus belongs to the space
\begin{equation}\label{eqn:space}
    \Pi := \left\{ \pi = \{\pi_{t}\}_{t=0}^{\infty} \;\middle|\; \pi_{t} : \mathcal{P}(S) \to \mathcal{H} \text{ is measurable} \right\}.
\end{equation}

Importantly, the input of the central controller is the space $\mathcal{P}(S)$, and the output is $\mathcal{H}=\mathcal{P}(A)^{|S|}$, which implies that the dimensionality of the input–output mapping is independent of the number of agents. In fact, previous work \cite{gu2021mean} shows that under Pareto-optimality criteria, cooperative MF-MARL can approximate cooperative MARL up to an error of $\mathcal{O}(\frac{1}{\sqrt{N}})$.

\section{Problem Formulation and Methodology}

Let $J$ be a collection of $n$ training or inference workloads (represented as integers in $[n]$). They arrive in an online fashion such that a workload $j \in J$ has a release time of $r_j$, a soft deadline of $d_j$, and $p_{j,k}$ is the processing time of workload $j$ on a MIG slice of size $k$, where $k \in \{1, 2, 3, 4, 7\}$. The processing times of the online jobs can be obtained from a job performance estimator such as MISO \cite{li2022miso}, but in this work, they are obtained from a probability distribution (along with the deadlines). 
There is a collection of $m$ MIG-enabled A100-40GB GPUs, indexed in $[m]$. The goal is to decide which job to allocate to which GPU and to which slice at any point in time (with preemption being permitted). Given the partitioning and scheduling strategy, we obtain the total energy consumption $e$ using the concave power characteristics of MIG on A100-40GB with a power cap of 250W. The precise values are $P=[40, 119, 160, 205.3, 243.9, 247.7, 248.5, 248.5]$, where the $i$-th entry in this $0$-indexed list approximately corresponds to $i$ MIG slices being utilized \cite{lipe2025energy}. The average tardiness $t$ is defined as $(\sum_{i=j}^n \max(C_j-d_j,0))/n$, where $C_j$ is the completion time of job $j$. This problem models a data center where jobs arrive in real time, and decisions must be made to repartition GPUs and schedule jobs across GPU partitions to optimize both energy and performance.

Most workloads consist of linear and sublinear (we consider capped jobs as sublinear) throughput jobs \cite{tan2021serving, zhang2023migperf} scaling with slice size, as shown in Figure \ref{fig:thoughput}. An effective GPU repartitioning and scheduling strategy can allocate sublinear jobs appropriately, avoiding oversized slices that waste resources and undersized slices that cause delays. To analyze SMART-MIG across different sublinearity levels of throughput, we define the degree of sublinearity for a group of jobs as the average of their throughputs across 1g, 2g, 3g, 4g, and 7g slices, effectively converting the group into a single ``average job". If the group were fully linear, its throughput curve would form a straight line with a slope equal to the throughput on a 1g slice; we denote this slope as $k_{linear}$. For any real group of jobs, we fit a line to its actual throughput curve by minimizing the mean squared error and denote the slope of this fitted line by $k_{sublinear}$. Clearly, for the same group of jobs, the slope of the fitted line is always smaller than the slope of the linearized throughput line. We denote $ss=\frac{k_{sublinear}}{k_{linear}}\in [0,1]$ as the sublinearity score for the group of jobs.

\begin{figure}[ht] 
    \centering 
    \includegraphics[width=0.7\linewidth]{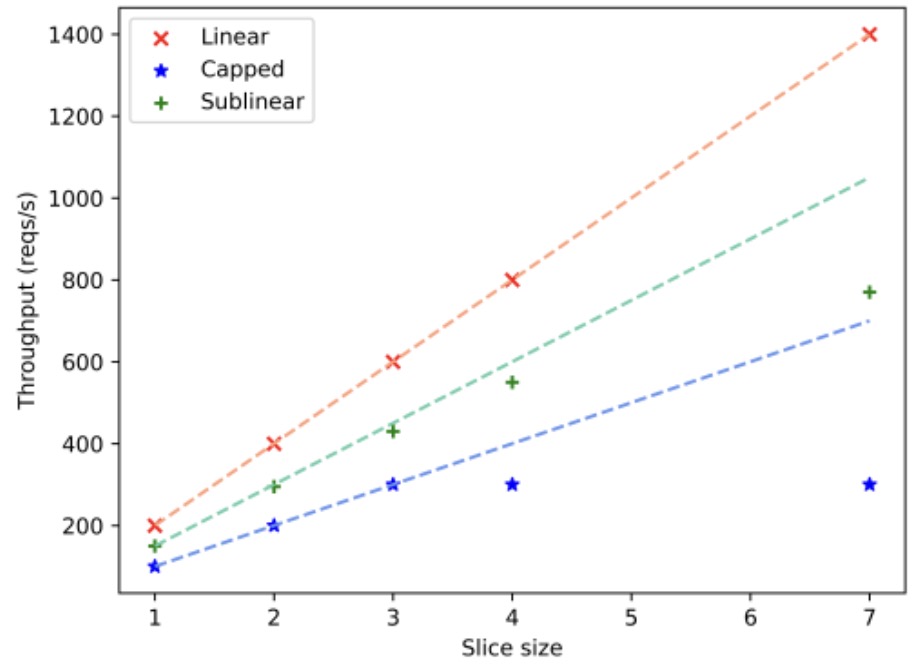} 
    \caption{Thoughput curves of different types of jobs} 
    \label{fig:thoughput} 
\end{figure}

The fact that the power consumption of MIGs increases concavely with the number of active slices \cite{lipe2025energy} motivates us to maximize GPU utilization to reduce energy costs. These insights demonstrate the feasibility of designing a system that optimizes both energy and tardiness. To this end, our aim is to minimize the $ET$ value across simulations, as defined in \cite{lipe2025energy} as:
\begin{equation} 
\label{eq:et}
ET = \frac{1}{N}\sum_{k=1}^N\frac{ae_k+t_k}{a+1},
\end{equation}
where $N$ is the number of simulations, $e_k,t_k$ are the total energy and the average tardiness of the simulation $k \in [N]$. Let $\bar{e}$ and $\bar{t}$ be the average of the total energy and average tardiness over $N$ simulations, respectively. $a$ is a scaling parameter which is usually set to $\frac{\bar{t}}{2\bar{e}}$ as in \cite{lipe2025energy}. 

\subsection{Tardiness-Energy Curves with Increasing Number of GPUs}\label{sec:tradeoff}
The first point to consider is the number of GPUs required to support the maximum \emph{load} on the GPU system while preventing resource waste. We experiment by incrementally increasing the number of GPUs. We observe the Pareto curves of $\bar{t}$ versus $\bar{e}$ as the number of GPUs increases at various arrival rates (Figure \ref{fig:pareto}) by scheduling jobs using any given algorithm. 
\begin{figure}[ht]\label{alg:algorithm1} 
    \centering 
    \includegraphics[width=0.9\linewidth]{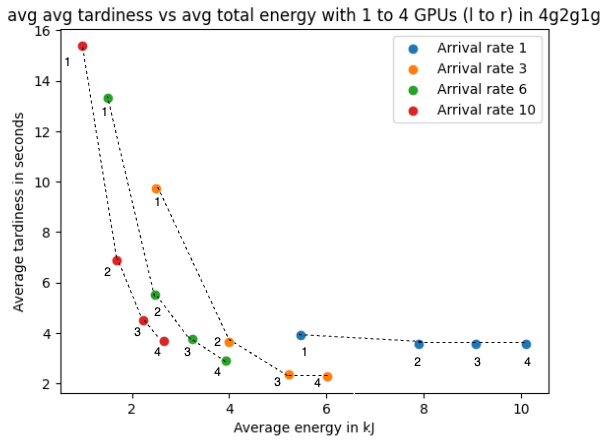} 
    \caption{Tardiness-Energy curves with varying arrival rates} 
    \label{fig:pareto} 
\end{figure}
As expected, we observe diminishing returns as the number of GPUs increases. In particular, as we continue to increase the GPU count, energy keeps increasing, but with smaller tardiness savings. The number of GPUs required for scheduling (which increases with increasing arrival rate) can be estimated from the knee of the plots (similar to the idea in \cite{kim2022paris}, but in a different context).

\subsection{SMART-MIG Overview}
The general workflow of SMART-MIG is as follows. Whenever a new job arrives, we first obtain its profile, which includes the expected completion time in slices of different sizes. In a typical workflow, a throughput profiler such as MISO \cite{li2022miso} can be used; here, we generate them from distributions derived from real-world data, as explained in Section \ref{sec:exp}. Next, when a new job arrives, or an existing job completes (or the controller may be invoked less frequently, for example, once every five arrival or completion events to reduce repartition overhead), the MF-MARL model deployed at the central controller reads the state distributions of jobs and MIGs and outputs reallocation instructions to guide MIG repartitioning. Once the configuration of each MIG is fixed, the problem reduces to a multi-machine scheduling problem, where we employ a carefully designed scheduling algorithm to assign jobs across different slices of different GPUs. We first present scheduling algorithms under static configurations in Section \ref{sec:algo}, followed by the central controller that directs dynamic MIG repartitioning in Section \ref{sec:controller}.

\subsection{Scheduling Algorithms}\label{sec:algo}
We design algorithms that leverage the job throughput characteristics and the concave power curve of GPUs when making assignment decisions.
Unlike the similar single-machine problem studied in \cite{lipe2025energy}, the presence of multiple GPUs significantly complicates scheduling decisions, necessitating the use of more intelligent scheduling algorithms. The pseudocodes of algorithms depicted here outline the choices of slices and GPUs, and not the entire scheduling algorithms. The jobs in each of them are ordered by earliest deadline first (EDF).

We consider the following algorithms:
\begin{enumerate}
    \item \textbf{EDF Most Packed (EDF)}: The jobs in the ready queue are sorted in increasing order of deadlines, and each job is assigned one at a time to the slowest slice of the most fractionally packed GPU that still finishes the job on time. If all the GPUs are occupied, then the job is assigned to the GPU slice with the fastest job completion time.
\begin{algorithm}[H]
\caption{EDF Most Packed (EDF)}
\begin{algorithmic}[1]
  \State Sort GPUs by occupied slices (desc), Jobs by deadline (asc)
  \For{job in Jobs, gpu in GPUs}
  \For{slice in free slices of gpu}
    \State $t \gets now + p_{job,slice}$
    \If{$t \leq d_{job}$} \State assign job to $(gpu,slice)$; \Return \EndIf
    \If{$t < best\_finish$} \State best\_choice $\gets (gpu,slice)$ \EndIf
  \EndFor
  \EndFor
  \State assign job to best\_choice (earliest finish)
\end{algorithmic}
\end{algorithm}
    \item \textbf{EDF Marginal ET (MET)}: The jobs in the ready queue are sorted in increasing order of deadlines, and each job is assigned to the GPU that has the lowest marginal $ET$ gain, where $ET$ is defined in equation \ref{eq:et}. The marginal $ET$ is calculated by assuming that no more jobs will be scheduled on that GPU.

\begin{algorithm}[H]
\caption{EDF Marginal ET}
\begin{algorithmic}[1]
  \State Sort Jobs by deadline (asc)
  \For{job in Jobs}
  \State tardiness $\gets \max(0, now + p_{job,slice} - d_{job})$
  \State marginal $\gets$ extra energy if the slice used until finish
  \State chosen $\gets \arg\min_{s \in \text{free slices}}\dfrac{a \cdot \text{marginal} + \text{tardiness}}{a+1}$
  \State assign job to chosen slice
  \EndFor
\end{algorithmic}
\end{algorithm}

    \item \textbf{Categorical EDF (CEDF)}: This algorithm classifies jobs into four categories, ranging from least sublinear to most sublinear. The GPUs are sorted in decreasing order of the average slice size. The jobs are first sorted by these categories and within each category, by deadlines. They are then packed GPU-wise into the slowest slice that completes the job on time, otherwise the fastest slice. The procedure for choosing the slice is similar to EDF Most Packed, but the GPUs are ordered by average slice size rather than fractional utilization.
        \begin{algorithm}[H]
\caption{Categorical EDF (CEDF)}
\begin{algorithmic}[1]
\Procedure{ClassifyJob}{job}
\State
$\begin{cases}
 \text{if } thr(4g) < 0.8 \cdot thr(7g), thr\_class = 1 \\
  \text{elif } thr(2g) < 0.8 \cdot thr(7g), thr\_class = 2 \\
  \text{elif } thr(1g) < 0.8 \cdot thr(7g),  thr\_class = 3 \\
  \text{otherwise},  thr\_class = 4
\end{cases}$
\EndProcedure
\State Sort GPUs by average slice size (desc)
\State Sort Jobs by $thr\_class$ and deadline (asc)
\State Assign Jobs using steps 2-13 of Algorithm 1
\end{algorithmic}
\end{algorithm}
\end{enumerate}

\subsection{MF-MARL based Central Controller}\label{sec:controller}
As the MIGs are identical, indistinguishable, and interchangeable, the key is to check whether a valid MIG configuration exists for a given task, and not which MIG executes it. The jobs share the same property: we are concerned with the characteristics of pending jobs, not their individual identities. So, for any job, as long as the GPU configuration is appropriate, the execution outcome remains the same regardless of which GPU processes it; likewise, given a fixed GPU configuration, jobs with the same characteristics yield identical results.

This symmetry motivates the use of Mean-Field MARL, where the state information of all GPUs and jobs can be represented as distributions. In practice, GPUs can take on 12 possible configurations, denoted $c_1,\cdots,c_{12}$. The current configuration of each MIG defines its state. We define this moment of arrival or completion of a job as step $t$. Assuming there are $m$ MIGs in total, the state of the $i^{th}$ MIG is denoted as $s^{i,cfg}_t$ for each $i \in [m]$. The state distribution of the MIGs can then be computed as:
$\mu^{MIGs}_{t}(c_q) = \frac{\sum_{i=1}^{m} \mathbf{1}\!\left( s_{t}^{i,cfg} = c_q \right)}{m}.$

The state of a job is defined by its duration and deadline, denoted for the $j^{th}$ job as $(s^{j,du}_t,s^{j,dl}_t)$ in step $t$. Since these states are continuous, they must be discretized in order to be represented as distributions. Specifically, let $du_1,\cdots,du_V$ denote the discretized intervals of job duration. Then $s^{j,du}_t=du_v$ if and only if the duration of job $j$ falls into the interval $du_v$. Similarly, let $dl_1,\cdots,dl_O$ represent the discretized intervals of job deadlines. To simplify the problem, we consider the job duration and deadline distributions separately. They can be computed as:
$\mu^{Du}_{t}(du_v) = \frac{\sum_{j=1}^{n} \mathbf{1}\!\left( s_{t}^{j,du} = du_v \right)}{n}$ and 
$\mu^{DL}_{t}(dl_o) = \frac{\sum_{j=1}^{n} \mathbf{1}\!\left( s_{t}^{j,dl} = dl_o \right)}{n}$, assuming that there are $n$ jobs, 

Following the MF-MARL paradigm described in Section \ref{sec:mf-marl}, the central controller takes the above distributions as input and outputs a policy $h_t$ applicable to all MIGs, guiding their repartitioning. For any given MIG, this policy takes its current configuration as input and returns a probability distribution over the next possible configurations. The next configuration of the MIG is then sampled from this distribution. The form of the policy can be expressed as follows:
\begin{equation}\label{eq:pmat}
h_t=
\begin{pmatrix}
p(c_{1}\mid c_{1}) & \cdots & p(c_{12}\mid c_{1}) \\
\vdots & \ddots & \vdots \\
p(c_{1}\mid c_{12}) & \cdots & p(c_{12}\mid c_{12})
\end{pmatrix},
\end{equation}
where $p(c_{q}\mid c_{p})=p(s_{t+1}^{i,cfg}=c_{q}\mid s_{t}^{i,cfg}=c_{p})$ for all $i \in [m]$.

After the central controller issues repartitioning instructions, the configurations of all MIGs are updated based on the transition distribution corresponding to their current state. We adopt top-$k$ sampling, a widely used technique in large language models, to improve the system stability. Specifically, top-$k$ sampling selects the $k$ configurations with the highest probabilities from the distribution, re-normalizes their probabilities, and then randomly samples the next configuration from this restricted set. The system then applies the scheduling algorithm to assign jobs across different slices, and the GPUs process these tasks until either a new job arrives or an existing job is completed. At this point, the central controller receives a reward defined as
\begin{equation}\label{eq:reward}
r_{t} = -\frac{ \text{tardiness}_{t} + a\times\text{energy\_consumption}_{t}}{1+a},
\end{equation}
where $\text{tardiness}_{t}$ and $\text{energy\_consumption}_{t}$ denote the new job delay and the energy cost during this processing interval. The parameter $a$ is set as discussed earlier. During all timesteps, the cumulative reward is set to $-ET$. Therefore, the reward of the central controller is directly related to the downstream objectives. This stepwise reward formulation alleviates the issue of reward sparsity, thereby facilitating effective learning.

Then, we use PPO to train the model to maximize the expected cumulative reward. We know that there is a 4-second repartitioning penalty \cite{li2022miso}, which is incorporated into the model indirectly by reducing the deadlines by this amount. 

\section{Lower Bounds}\label{sec:lower_bounds}

In this section, we show lower bound results for energy and tardiness. We first \emph{linearize} the sublinear jobs to their linear counterparts and then show that running each of them on $7g$ MIGs yields lower bounds for both energy and tardiness. 

For the energy lower bound, we reduce our problem to a polynomial-time solvable makespan minimization problem, and for the tardiness lower bound, we use a mixed integer program, which is numerically approximated (it is NP-hard).

\begin{definition}
    For a job $j$ with processing time $p_{j,1}$ on $1g$, we define its linearization as a job $L(j)$ with processing times $p_{j,i} = p_{j,1}/i$ for each $i \in \{2,3,4,7\}$.
\end{definition}

For a collection of jobs $J$, let $L(J)$ represent the corresponding collection of linearized jobs.

\subsection{Energy Lower Bound}

\begin{lemma}\label{lem:linear}
    Any schedule $\Pi$ with a combination of linear and sublinear throughput jobs $J$ can be converted to a corresponding schedule $\Pi'$ on the linearized job set $L(J)$ without increasing the energy.
\end{lemma}

\begin{proof}
    The above result follows from the power curve's concavity and the throughput curves' sub-linearity, implying that the jobs' linear extensions will consume less energy running on a corresponding schedule with the same starting times as the original schedule (and will finish no later).
\end{proof}

From now on, we will consider all jobs to be linear and let the energy used by a schedule $\Pi'$ on a linearized job set $L(J)$ be $E^{\Pi'}(L(J))$.

\begin{lemma}\label{lem:single_machine}
    Any schedule $\Pi'$ on a linearized job set $L(J)$ on $m$ machines with any configuration can be converted to a schedule $\Pi''$ on $m$ $7g$ machines without increasing energy.
\end{lemma}

\begin{proof}
    Consider the first machine without loss of generality, with $\Pi'$ restricted to it being denoted as $\Pi'_1$. Let the finishing time of $\Pi'$ be $T$. For ease of calculation, we deduct $P_0T$ from the total energy $E^{\Pi'_1}(L(J))$ and account for it later.

    We divide the schedule $\Pi'$ on the first machine into $k$ segments, where a new segment begins whenever a job starts or completes (or there is a reconfiguration). Consider a segment $i \in [k]$, and let its duration be $t_i$ and starting time and finish times be $s_i$ and $f_i$ respectively. Let there be $l$ jobs running on slices of sizes $y_{ij}$ for each $j \in [l]$, such that the $\sum_j y_{ij} = S_i$. The power consumed by these jobs is $P_{S_i} t_i$. Since idle power must be deducted, active power during this time is $(P_{S_i}-P_0) t_i$. Construct a corresponding segment on $7g$ where the parts of the jobs in segment $i$ are scheduled consecutively, one after the other. Its duration will be $\frac{S_i t_i}{7}$ which is no greater than $t_i$. The active power consumed will be $(P_7-P_0) \frac{S_i t_i}{7} = P_7 t_i\frac{S_i}{7} +P_0t_i\left(1-\frac{S_i}{7}\right) - P_0 t_i \leq (P_{S_i}-P_0) t_i$, using the concavity of the power curve. Finally, since each new segment has no greater duration than the original segment, the reduction is feasible and cannot result in an increase in the active energy when these segments are combined. The idle power is $P_0T$. Since we are building a total schedule $\Pi''_1$ of duration $T$, it can be added in the end. Finally, repeating this process for all $m$ machines proves the lemma.
    \end{proof}

    \begin{lemma}
        To compute a lower bound on the total energy, we calculate the energy of a minimum makespan schedule using jobs $L(J)$ on $m$ $7g$ machines.
    \end{lemma}

\begin{proof}
     It can be shown that for jobs $L(J)$ running using any schedule $\Pi$ on $m$ $7g$ machines, the total energy consumption is given by adding the active energy to the idle energy:
\begin{equation}\label{eq:energy_lb}
    E^\Pi(L(J))=\sum_{j\in J}\frac{p_j}{7} P_7+(mC^\Pi-\sum_{j\in J}\frac{p_j}{7})P_0,
\end{equation}
where $[n]$ is the index set of jobs and $C^\Pi$ is the makespan of a given schedule $\Pi$. 
Replacing $C^\Pi$ with the minimum makespan $C^*$ in (\ref{eq:energy_lb}) yields the energy lower bound (the active energy cannot be decreased, and minimizing the makespan minimizes the idle energy).
Note that $C^*$ can be computed in $O(n^2)$ time using the staircase algorithm for makespan minimization with preemption on uniform machines from \cite{labetoulle1984preemptive}.
\end{proof}

\subsection{Tardiness Lower Bound}

It can be shown that Lemmas \ref{lem:linear} and \ref{lem:single_machine} also hold for tardiness (using similar proofs). In this subsection, we replace $J$ with $L(J)$ and derive a lower bound result on the tardiness of $L(J)$. We frame our problem as a parallel identical-machine scheduling problem and develop a mixed integer program (MIP) with a just-in-time formulation. 

Let there be $m$ machines and $n$ linearized jobs. For each job, $j \in [n]$, let the release time be $r_j$, the processing time be $p_j$, and the deadline be $d_j$. After setting an $\epsilon$ that indicates the time step size, we set the maximum time horizon at $\tau = \lfloor (\max_j r_j+\sum_j p_j)/\epsilon \rfloor$. The time horizon is discretized into time slots from $0$ to $\tau$ with time slot $s$ corresponding to $[s\epsilon, (s+1)\epsilon)$.

\begin{itemize}
    \item Let $x_{j,s} \geq 0$ be the amount of job $j$ processed in time slot $s$ for each $j \in [n]$ and for each $s \in [0,\tau]$.
    \item Let $y_{j,s} \in \{0, 1\}$ be a binary variable that is $1$ if job $j$ has been completed until time slot $s$, and $0$ otherwise.
    \item Let $C_j$ be the completion time of job $j$ for each $j \in n$.
    \item Let $T_j = \max\{0, C_j - d_j\}$ be the tardiness of job $j$ for each $j \in [n]$.
\end{itemize}

\textbf{Objective}: $\min \sum_{j=1}^n T_j$ 

\textbf{Constraints:} 
\begin{enumerate}
    \item Processing requirement: $\sum_{s=0}^\tau x_{j,s} \leq p_j \quad \forall j \in [n]$.
    \item Machine capacity: $\sum_{j=1}^n x_{j,s} \leq m \epsilon \quad \forall s \in [0, \tau]$.
    \item Job processing rate: $x_{j,s} \leq \epsilon \quad \forall j \in [n], \forall s \in [0, \tau]$.
    \item Release time: $x_{j,s} = 0 \quad \forall j, \forall s < \lfloor r_j /\epsilon \rfloor$.
    \item Completion indicator: $y_{j,s} \leq y_{j,s+1} \quad \forall j, \forall s \in [0, \tau - 1]$ and $y_{j,t} \leq \frac{\sum_{s=0}^t x_{j,s}}{p_j} \quad \forall j, \forall t \in [0, \tau]$.
    \item Completion time: $C_j = \epsilon \sum_{s=0}^\tau (1 - y_{j,s}) \quad \forall j \in [n]$.
    \item Tardiness: $T_j \geq C_j - d_j, ~T_j \geq 0 \quad \forall j \in [n]$.
    \item Domain constraints: $y_{j,s} \in \{0, 1\}, \quad x_{j,s} \geq 0 \quad \forall j \in [n], \forall s \in [0, \tau]$.
\end{enumerate}

The constraints ensure $y_{j,t}$ is non-decreasing and switches from $0$ to $1$ when $\sum_{s=0}^t x_{j,s} = p_j$ because the tardiness minimization in the objective pushes the indicator to be 1 as soon as possible (if the completion time is after the due date). $C_j = \epsilon \sum_{t=0}^\tau (1 - y_{j,t})$ counts the number of slots where $y_{j,t} = 0$, multiplied by the slot width, equal to $\epsilon$ times the smallest $t$ where $y_{j,t} = 1$, which approximates when job $j$ is completed. This procedure uses $O(n\tau)$ variables and constraints, and lower-bounds job tardiness with an error up to twice the time slot width (errors arising from both release time and completion time discretization). Hence, the lower bound can underestimate the actual tardiness by at most $2n \epsilon$ when considering $n$ jobs. 

However, in practice, for both energy and tardiness, there is additional error from the jobs being sublinear in throughput. But since most common inference jobs are observed to be of linear throughput \cite{tan2021serving}, we observe that we are empirically not too far from both the lower bounds (refer to Section \ref{sec:exp}).

\section{Experimental Results}\label{sec:exp}
Here, we explain the processes that generate our workloads and the parameters that we vary to perform experiments on MIGs to demonstrate improvements in energy and tardiness.

\subsection{Experimental Settings}
The jobs and machines are as defined previously. For the MIG experiments without repartitioning, we assume that the workload release times are simulated from a fixed-rate Poisson arrival process. We then split them into training and inference jobs. The processing time characteristics of the workloads are generated from a lognormal distribution \cite{li2004workload} for training jobs and an exponential distribution for inference jobs \cite{cruzes2025revolutionizing}, with their slicewise processing times determined by using the Resnet-50 and BERT-Base throughputs from \cite{zhang2023migperf} and \cite{tan2021serving}. Since deadline information is not generally available in data center traces, we model it as $d_j \sim Unif(1, 1.5) * p_{j, 7}$, which is similar to the deadline distribution in \cite{zhang2025deadline}.


For our repartitioning experiments, we generate workloads with a base arrival rate that varies during the day \cite{weng2022mlaas,lipe2025energy} as in Figure \ref{fig:Variation of Arrival Rate Throughout 24-Hour Day}, which is multiplied by a multiplier. Next, we split each of these workloads into training and inference types in a fixed proportion, and again consider their throughput characteristics, as in the no-repartitioning case, based on the experiments in \cite{tan2021serving, espenshade2024characterizing, zhang2023migperf}.
 \begin{figure}[ht] 
     \centering 
     \includegraphics[width=0.9\linewidth]{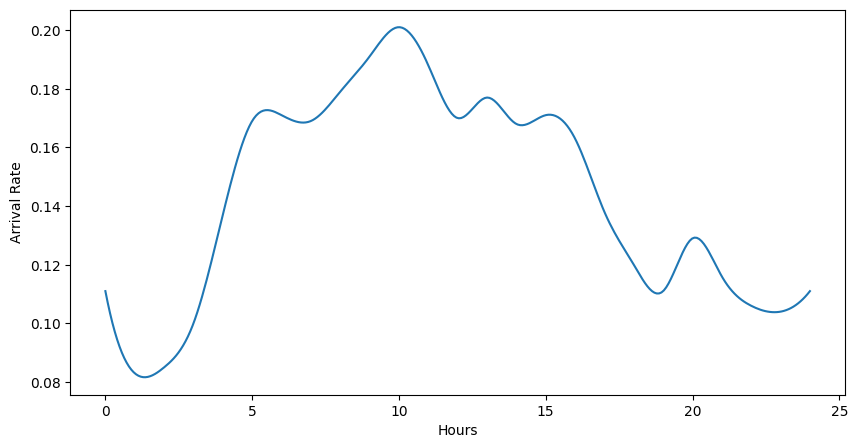} 
     \caption{Arrival rate variation during a day} 
     \label{fig:Variation of Arrival Rate Throughout 24-Hour Day} 
 \end{figure}
 

We do not consider distributed sliced training because, according to the NVIDIA documentation \cite{nvidiaUserGuide}, MIG-enabled GPUs do not have P2P communication, which is necessary for such tasks, and it is beyond our scope.

Using the described experimental setup, we evaluate the performance gains of MIGs compared to standard GPUs without MIG. To emulate multi-tenant, large-scale workloads with high job counts and short-lived tasks, we set the job arrival rate to be $20$ times the base rate. Under such conditions, non-partitioned GPUs fail to keep up, while MIGs achieve $60-100 \times$ lower tardiness at high arrival rates (Table \ref{tab:no_mig}), since most training jobs have sublinear throughput curves that degrade efficiency and cause severe backlog. Although the performance of non-MIG GPUs can be partially improved by dedicating some GPUs to inference and others to training (still yielding a $3-4 \times$ improvement), their performance remains significantly worse than that of MIG-enabled GPUs. These results strongly motivate our focus on MIGs.

\begin{table}[h]
\centering
\caption{Performance comparison under different arrival rates}
\label{tab:no_mig}
\begin{tabular}{|l|l|c|c|c|c}
\hline
\textbf{} & \textbf{} & \textbf{20x} & \textbf{16x} & \textbf{12x} \\
\hline
\multirow{2}{*}{Energy} 
  & SMART-MIG      &  2.5M & 2.3M  & 2M   \\
  & No Partition    & 3.1M & 2.4M  & 2M    \\
\hline
\multirow{2}{*}{Tardiness} 
  & SMART-MIG       & 1.0037 &  0.44 & 0.18  \\
  & No Partition    & 96.67 &  16.58& 0.34   \\
\hline
\end{tabular}
\end{table}

We compare the energy and tardiness from our repartitioning algorithm with the best scheduling algorithm in practice, as well as with the lower bound, and with the naive results obtained by scheduling on GPUs with fixed configurations.

\subsection{Results for Scheduling Algorithms without Repartitioning}
We first evaluate different scheduling algorithms without repartitioning by running experiments with varying (but fixed per run) arrival rates and configurations, aiming to characterize their relative performance, strengths, and suitable scenarios. Each job queue contains $1000$ jobs with a $4:1$ inference–training job ratio. We then simulate a full day of jobs (as in Section VI.C) to identify the best-performing algorithm, which is subsequently used to train the central controller model. Experiments are conducted on 8 GPUs without penalizing preemption. We compare the three algorithms from Section \ref{sec:algo} against a simple ‘Random’ baseline that assigns the earliest-deadline job to a randomly available slice.

\begin{itemize}
\item \textbf{Effect of varying arrival rates:} At lower arrival rates with smaller queues, we generally see Most Packed EDF performs well in terms of energy (Figure \ref{fig:arrivalrate1}), whereas at higher arrival rates, Marginal ET performs the best (Figure \ref{fig:arrivalrate8}). The energy is about $100\%$ worse than the lower bound at arrival rate 1, whereas it is only $20-30\%$ worse for the others, mainly because the loose lower bound introduces larger errors and we cannot pack as many GPUs efficiently at such a low arrival rate.

\begin{figure}[ht] 
    \centering 
    \includegraphics[width=\linewidth]{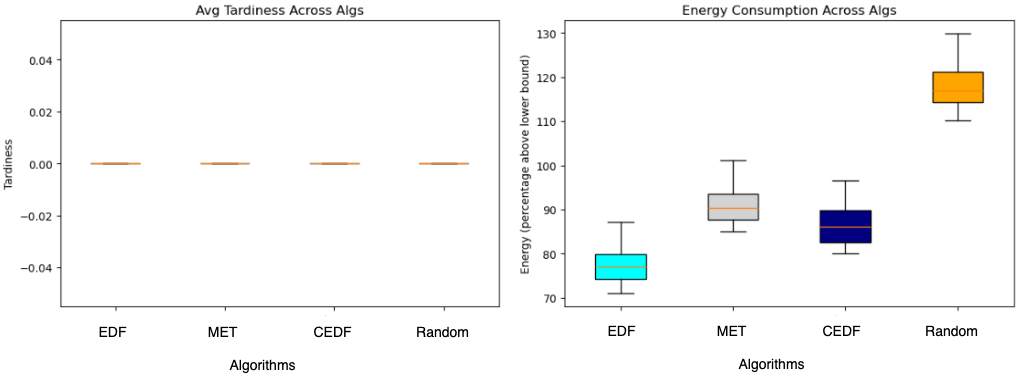} 
    \caption{Average tardiness and energy consumption at arrival rate 1} 
    \label{fig:arrivalrate1} 
\end{figure}

\begin{figure}[ht] 
    \centering 
    \includegraphics[width=\linewidth]{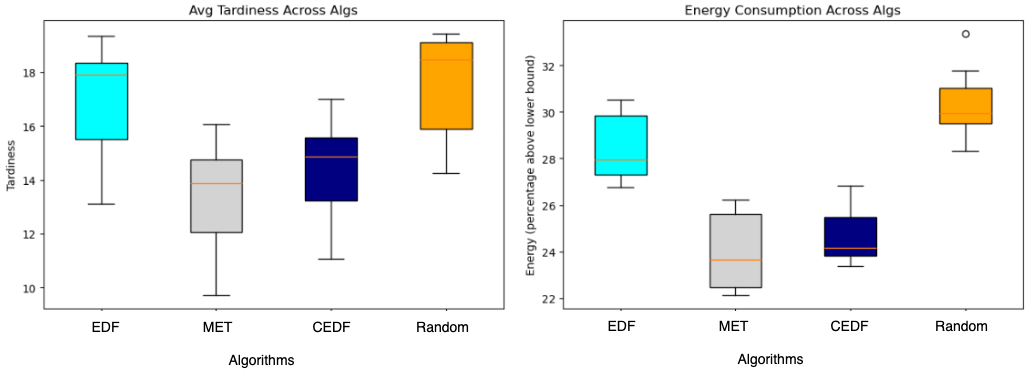} 
    \caption{Average tardiness and energy consumption at arrival rate 8} 
    \label{fig:arrivalrate8} 
\end{figure}

\item \textbf{Effect of varying the configuration:} This change is most prominent for CEDF as it sorts GPUs by average slice size (Figure \ref{fig:config1}), and hence benefits from having partitions with all slices having similar sizes. With such `good' configurations, CEDF performs the best in terms of energy and tardiness.

\begin{figure}[ht] 
    \centering 
    \includegraphics[width=\linewidth]{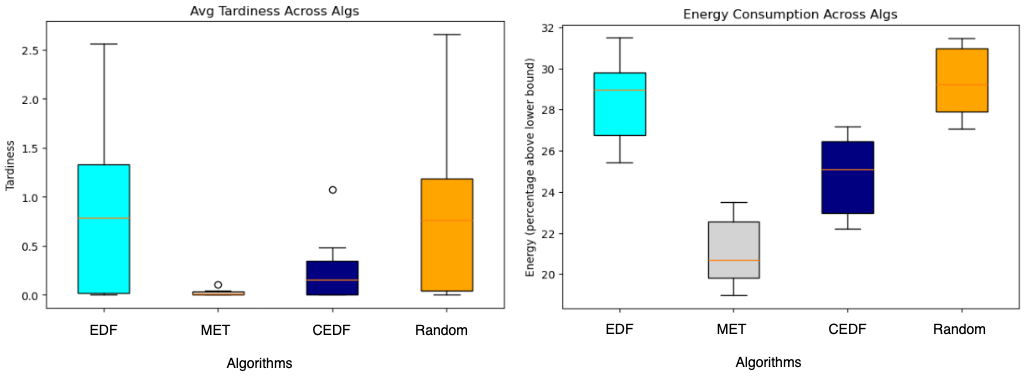} 
    \caption{Average tardiness and energy consumption for configurations where slice size varies within a GPU} 
    \label{fig:config1} 
\end{figure}

\begin{figure}[ht] 
    \centering 
    \includegraphics[width=\linewidth]{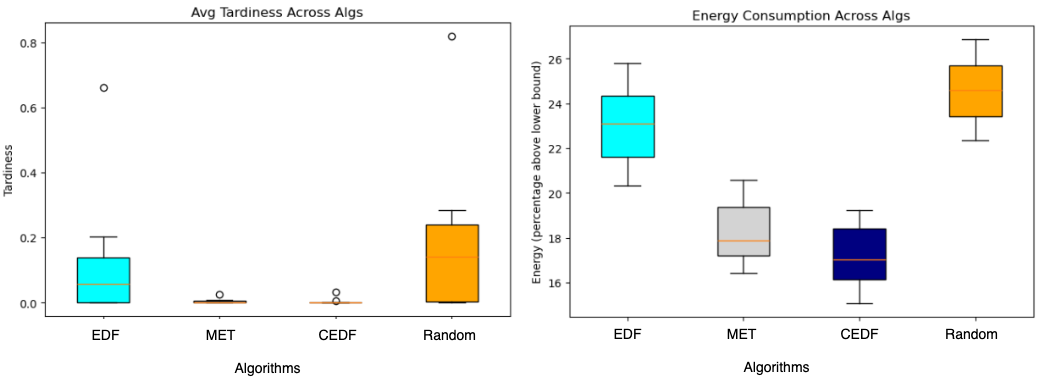} 
    \caption{Average tardiness and energy consumption for configurations where slice sizes are similar within a GPU} 
    \label{fig:config2} 
\end{figure}

\item \textbf{Results on realistic 24-hour queues:} These queues have varying arrival rates through the queues, and are also longer, with over 4000 jobs. Here, we find that CEDF generally performs best (Figure \ref{fig:realistic1}),  even when tested on varying configurations.

\begin{figure}[ht] 
    \centering 
    \includegraphics[width=\linewidth]{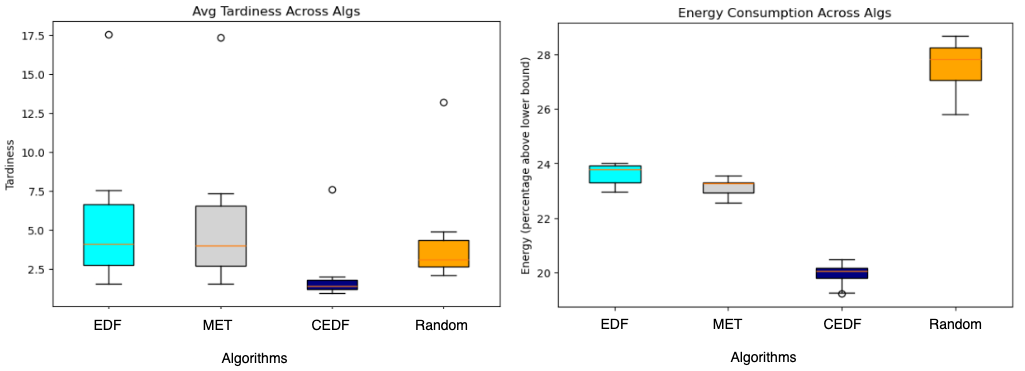} 
    \caption{Average tardiness and energy consumption for realistic 24-hour queues} 
    \label{fig:realistic1} 
\end{figure}
\begin{figure}[ht] 
    \centering 
    \includegraphics[width=\linewidth]{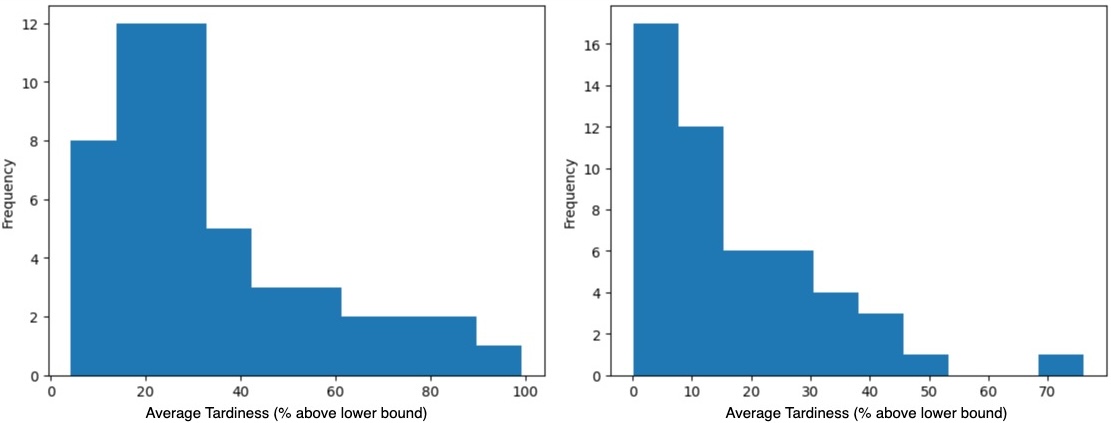} 
    \caption{Average tardiness compared with lower bounds} 
    \label{fig:Average tardiness vs lower bounds on 2g2g3g} 
\end{figure}

\item \textbf{Tardiness lower bound results:} Since computing the tardiness lower bound is computationally complex, we present tardiness results for running 50 simulations of 50 mixed throughput jobs and 50 linear throughput jobs, each using MP-EDF on $4$ GPUS with configuration 8 and arrival rate 8 (Figure \ref{fig:Average tardiness vs lower bounds on 2g2g3g}). The average tardiness is above the lower bound by $\sim 34\%$ for mixed throughput jobs and $\sim 17\%$ for linear throughput jobs, on average.

\end{itemize}
\subsection{Results for SMART-MIG}
In our evaluation, we consider a node with 8 MIGs and set the job arrival rate to 20 times the base rate. Throughout these experiments, the ratio of training-to-inference jobs is fixed at $1:4$. For SMART-MIG, the central controller is invoked every 5 steps with top-3 sampling (as explained in Section \ref{sec:controller}). Once the configurations are fixed, we use CEDF, which generally yields the best energy–tardiness trade-off among the algorithms we tested. The central controller model is trained on  $\sim 5M$ jobs, equivalent to 5000 days of workload at this arrival rate. We select the static configuration [1, 1, 2, 3, 3, 5, 5, 10] as a baseline for comparison with SMART-MIG since we find that CEDF performs best with it.

Figure \ref{fig:ppo_result} shows that SMART-MIG reduces average tardiness by $25\%$ compared to CEDF with a static configuration and by $40\%$ compared to EDF with a static configuration. In terms of energy consumption, SMART-MIG achieves improvements of $1.2\%$ over CEDF and $7\%$ over EDF. Evaluated with the $ET$ metric, SMART-MIG delivers overall gains of $18\%$ and $32\%$ relative to CEDF and EDF, respectively. It is important to note that in this study, SMART-MIG was trained with CEDF as the scheduling algorithm, though it can, in principle, be paired with any mature multi-machine scheduler. The observed improvement of SMART-MIG over its static configuration counterpart stems from its ability to dynamically repartition MIGs based on the evolving job queue, allowing the system to better handle varying urgency and throughput demands. In the figure, the red line marks the theoretical lower bound of average energy consumption. Comparing this lower bound with our three evaluated methods, we find that RL is $27\%$ worse and CEDF is approximately $29\%$ worse in terms of energy efficiency. Recall that this lower bound is not tight, as the calculation involves a linearization of jobs, which underestimates the true energy lower bound. Moreover, our focus is on multi-objective optimization, aiming to balance energy consumption with job tardiness. Under these considerations, we argue that all three algorithms achieve reasonably effective energy optimization.

We also notice that SMART-MIG cannot substantially reduce energy consumption via repartitioning. A plausible explanation is: given the requirement to maintain low job tardiness, the baseline CEDF without repartitioning is already close to optimal and leaves little room for further improvement. However, the improvement in tardiness is quite large.

\begin{figure}[ht] 
    \centering 
    \includegraphics[width=\linewidth]{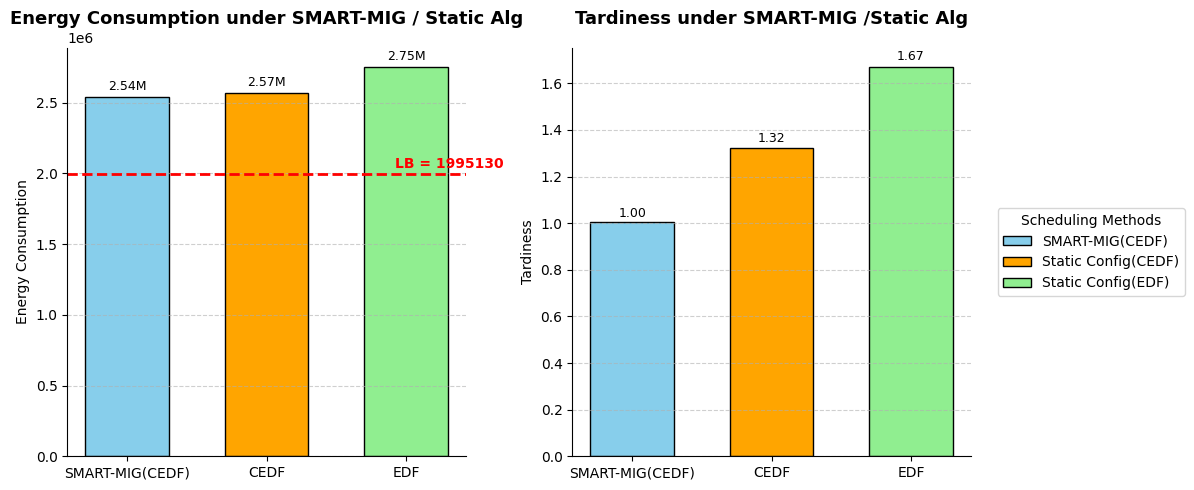} 
    \caption{Energy consumption and average tardiness results} 
    \label{fig:ppo_result} 
\end{figure}

\begin{figure}[ht] 
    \centering 
    \includegraphics[width=\linewidth]{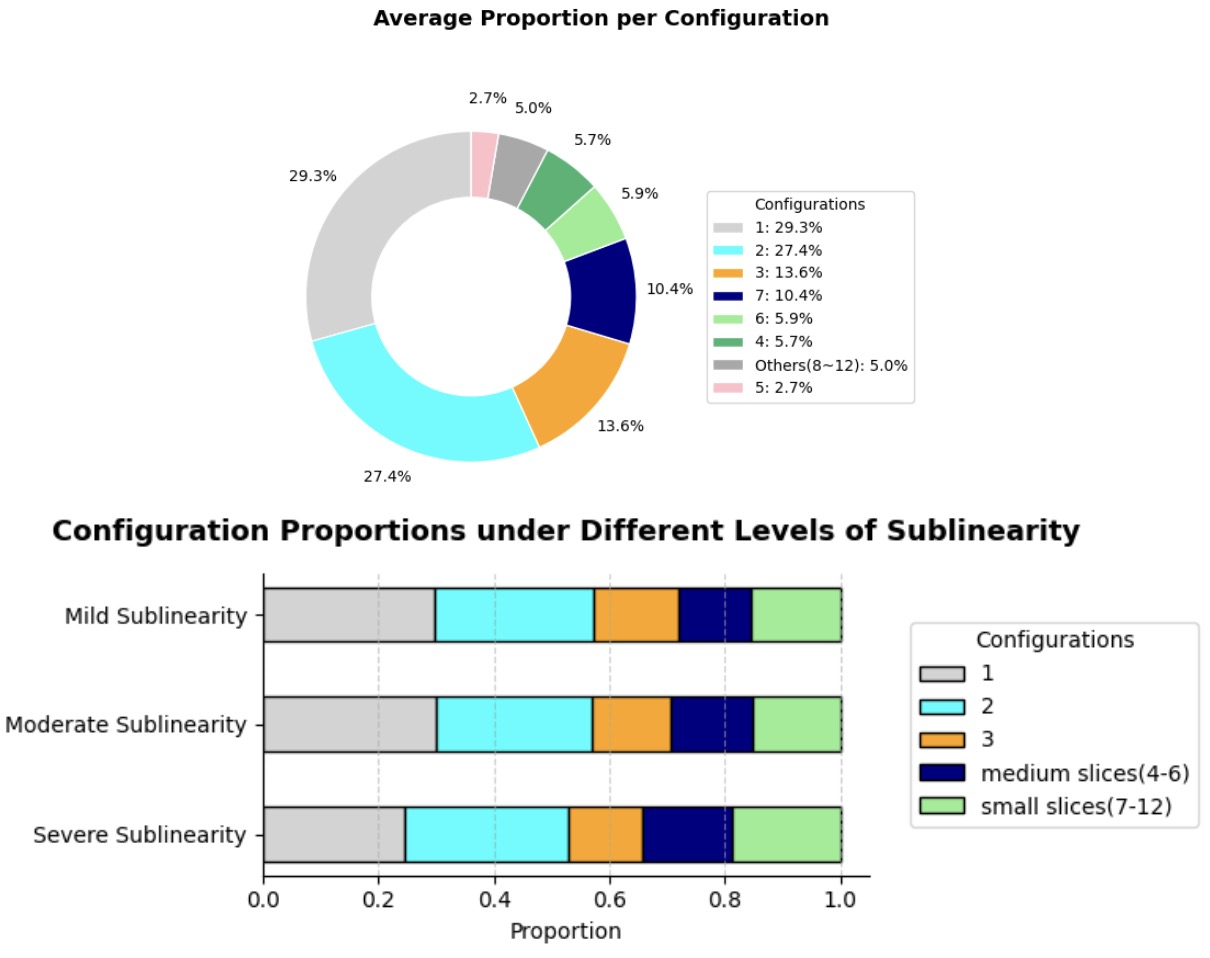} 
    \caption{Configuration characterization while repartitioning} 
    \label{fig:behavior} 
\end{figure}

Next, we analyze the repartitioning behavior of SMART-MIG. It is clear from \cite{lipe2025energy} that if all jobs were perfectly linear, then MIGs would provide no benefit. MIGs become meaningful only when sublinear jobs are present, since in such cases, the system must balance energy and tardiness. Thus, effective repartitioning essentially amounts to intelligently scheduling jobs with different throughput curves.

In our experiments, the sublinearity scores (as defined in section IV) of the job groups encountered by SMART-MIG at different steps range from $0.67$ to $0.84$. Based on this range, we categorize job groups into equally divided intervals:
\begin{itemize}
    \item Severe Sublinearity: $ss\in[0.67, 0.72]$
    \item Moderate Sublinearity: $ss\in[0.72, 0.78]$
    \item Mild Sublinearity: $ss\in[0.78, 0.84]$
\end{itemize}

The central controller encounters job groups with varying degrees of sublinearity across different experiments and timesteps. For each of these three categories, we measure the average repartitioning response of the system, namely the distribution of selected configurations. For comparison, we also report the global average distribution of configurations over all cases (Figure \ref{fig:behavior}).

The most frequently selected configurations are 1, 2, 3, and 7. Meanwhile, several configurations containing small slices (e.g., 8–12) collectively account for $5\%$, demonstrating that a reasonable GPU configuration should maintain a balance among large, medium, and small sliced GPUs. The specific balance should adapt dynamically to variations in job arrival rate and throughput. Furthermore, as the central controller encounters an increasing number of sublinear jobs, it tends to select configurations with larger indices that contain more small slices. This strategy increases task parallelism without undermining overall job execution efficiency, thereby reducing latency. Such variations do not cause substantial deviation in the average proportion of each configuration, indicating stability in job processing.

Put differently, SMART-MIG smartly balances harder-to-complete jobs (those more prone to incurring delays) with easier workloads by dynamically adjusting the configurations.

\section{Conclusion and Future Work}
We propose SMART-MIG, a scalable MIG-based parallel computing framework that jointly optimizes job performance efficiency and data center energy consumption. To enable fair system evaluation, we establish theoretical lower bounds on energy and tardiness. Our experiments first demonstrate the superiority of MIGs over traditional GPUs under machine learning workloads, where the latter nearly break down (with 60-100x worse tardiness). We also design a scheduling algorithm (CEDF) that leverages the nonlinearity of the job and power curves, achieving performance within $\sim27\%$ of the energy lower bound. Our dynamic repartitioning approach based on MF-MARL, when combined with CEDF reduces the tardiness further by $\sim25\%$ while also decreasing the energy consumption. These results highlight the effectiveness of using MF-MARL with intelligent scheduling for achieving energy-tardiness gains in MIG-based systems. 

There are several promising future directions. Some of them include (i) applying RL for job scheduling in conjunction with repartitioning, (ii) further tightening the lower bounds (by avoiding linearization), (iv) extending our work to heterogeneous GPUs, and (iv) exploring geometrical job alignment within and across GPUs to minimize heat dissipation and the energy needed for cooling.

\section*{Acknowledgment}
This research is partially supported by funding from IBM. The authors thank Ellie Lipe for her contributions in the initial part of this project.

\bibliographystyle{ieeetr}
\bibliography{cite}

\twocolumn[%
{\begin{center}
\Huge
Appendix    
\end{center}}
]


\appendixAD

\section{Overview of Contributions and Artifacts}

\subsection{Paper's Main Contributions}

\begin{description}
\item[$C_1$] SMART-MIG framework (integrating ML and OR):
The paper proposes SMART-MIG, a large-scale MIG scheduling framework with two components: (a) using Mean-Field Multi-Agent Reinforcement Learning (MF-MARL) with Top-k sampling to repartition GPUs, and (b) designing EDF-based scheduling algorithms (without repartitioning) to assign jobs to MIG slices across multiple GPUs by leveraging job throughput characteristics and MIG’s concave power curve.
\item[$C_2$] Energy and tardiness lower bounds: The paper derives energy and tardiness lower bounds as references for evaluating the scheduling policies.
\item[$C_3$] Extensive experimental validation:
The paper conducts extensive experiments of the studied algorithms to validate effectiveness, including: performance deterioration of the traditional no-partition GPU model under high arrival rates; static repartitioning operating within $\sim 30\%$ of the energy and tardiness lower bounds on average; MF-MARL–based dynamic repartitioning yielding an additional $\sim 18\%$ improvement in energy–tardiness efficiency; and consistent advantages across diverse workloads (e.g., even at low arrival rates, reducing tardiness by $\sim 47\%$ vs. the no-MIG baseline).
\end{description}

\subsection{Computational Artifacts}
\begin{description}
\item[$A_1$] Scheduling \& Simulation Artifact
\item[$A_2$] MF-MARL Dynamic Repartitioning Artifact (Controller Training \& Inference)
\item[$A_3$] Lower-Bound Computation Artifact (Energy LB \& Tardiness LB)

\end{description}
DOI: https://doi.org/10.5281/zenodo.18663599
\begin{center}
\begin{tabular}{rll}
\toprule
Artifact ID  &  Contributions &  Related \\
             &  Supported     &  Paper Elements \\
\midrule
$A_1$   &  $C_1, C_3$ & Table 1 \\
        &        & Figure 6-13\\
\midrule
$A_2$   &  $C_1,C_3$ & Table 1 \\
        &        & Figures 6-13\\
\midrule
$A_3$   &  $C_2$ & Figures 11-12 \\
\bottomrule
\end{tabular}
\end{center}

\section{Artifact Identification}

\newartifact

\artrel
Reproduces the paper’s scheduling experiments for MIG-based multi-GPU job scheduling using the paper’s EDF-based heuristics and workload model.
\artexp
At lower arrival rates, Most Packed EDF tends to be good on energy; at higher arrival rates, MET tends to perform best; and on realistic 24-hour queues, CEDF largely performs best and is used downstream for SMART-MIG training.

\arttime
The reproduction time is mainly dominated by the artifact execution time; typically, a single simulation takes about 20 minutes.

\artin

\artinpart{Hardware} Any computer equipped with a CPU.
\artinpart{Software} The Python package stable-baselines3.
\artinpart{Datasets / Inputs}
Energy model uses the paper’s concave power values (A100-40GB, $250 W$ cap):
$P=[40,119,160,205.3,243.9,247.7,248.5,248.5]$. Release times follow a fixed-rate Poisson arrival process. Training job base durations follow a lognormal distribution. Inference job durations follow an exponential distribution. Deadlines are modeled as $d_j\sim Unif(1,1.5) p_{j,7}$.
\artcomp
Run simulation batches that mirror the paper’s structure:
\begin{itemize}
    \item Vary arrival rates (fixed per run), evaluate EDF / MET / CEDF (and Random baseline if desired) on the same workload draws. 
    \item Vary fixed GPU MIG configurations and compare algorithm sensitivity (the paper highlights configuration effects, especially for CEDF).
    \item Run realistic 24-hour queues with varying arrival rates through the day and longer queues (paper: “over 4000 jobs”). 
\end{itemize}
For each run, record:
\begin{itemize}
    \item Total energy $e$ (from the power curve aggregation).
    \item Average tardiness.
    \item ET objective.
\end{itemize}

\artout
Compare algorithms on (energy, tardiness) trade-offs across arrival rates and configurations; identify which algorithm is best under each regime (low rate vs high rate; 24-hour queues). Confirm that CEDF is a strong overall choice on realistic 24-hour queues and thus is a justified scheduler for the SMART-MIG controller experiments.

\newartifact

\artrel
Reproduces the paper’s SMART-MIG dynamic repartitioning results using Mean-Field MARL as a central controller that periodically changes the GPU partition distribution and cooperates with the scheduler
\artexp
\begin{itemize}
    \item SMART-MIG reduces average tardiness by $\sim 25\%$ vs CEDF static, and $\sim 40\%$ vs EDF static.
    \item SMART-MIG improves energy by $\sim1.2\%$ vs CEDF and $\sim7\%$ vs EDF.
    \item Overall ET metric improvement: $\sim 18\%$ vs CEDF and $\sim 32\%$ vs EDF.
\end{itemize}

\arttime
Training corresponds to processing $5M$ jobs (about 5000 simulated days) to learn the MF-MARL controller policy. It takes about 1200 minutes.

\artin

\artinpart{Hardware} High-performance CPU or possibly an entry-level GPU.
\artinpart{Software} The Python package stable-baselines3.
\artinpart{Datasets / Inputs} This part is the same as $A_1$.
\artcomp
Train the central controller:
\begin{itemize}
    \item Run the SMART-MIG loop where each step is a job arrival or completion event; every 5 steps, invoke controller, sample configuration updates with top-k sampling, then schedule jobs using CEDF
    \item Continue training until you reach the paper’s scale.
\end{itemize}

Evaluate SMART-MIG vs static baselines: Fix the evaluation workload setting. Compare against static scheduling baselines, including the paper’s chosen static configuration.

\artout
Verify the headline improvements (tardiness, energy, ET) relative to static EDF/CEDF baselines. Characterize configuration selection frequencies and how they shift with sublinearity severity (the paper reports most frequent configurations and qualitative trends).

\newartifact

\artrel
Reproduces the paper’s theoretical reference baselines by computing lower bounds for energy and tardiness used for benchmarking scheduling policies.
\artexp
\begin{itemize}
    \item Energy lower bound comparison: RL (SMART-MIG) reported $\sim 27\%$ worse than energy LB, and static CEDF $\sim 29\%$ worse, noting the lower bound is not tight due to linearization.
    \item The paper’s reported benchmarking example: average tardiness above the lower bound by $\sim 34\%$ for mixed throughput jobs and $\sim 17\%$ for linear throughput jobs in a specific evaluation setup.
\end{itemize}

\arttime
Typically, a single execution takes about 5 minutes for computing the energy lower bound, and 30 minutes for computing the tardiness lower bound.

\artin

\artinpart{Hardware} This part is the same as $A_1$.
\artinpart{Software} stable-baselines3 and ortools.
\artinpart{Datasets / Inputs} This part is the same as $A_1$.

\artcomp
\begin{itemize}
    \item Energy lower bound: compute energy of a minimum makespan schedule for job queue $L(J)$ on $m$ 7g machines. The minimum makespan can be computed in $O(n^2)$ using the staircase algorithm for preemptive makespan minimization on uniform machines, but we use a flow-based approach for an arbitrarily close approximation.
    \item Tardiness lower bound: solve the paper’s MIP using Google's OR-Tools.
\end{itemize}

\artout
Report LB values alongside algorithm outcomes (A1/A2), emphasizing the paper’s caution that LBs may be non-tight due to linearization and discretization. Reproduce the paper-style comparison plots/tables: algorithm energy vs energy LB; algorithm tardiness vs tardiness LB; and the derived ET comparisons. Compare the energy and tardiness to the theoretical lower bounds and confirm the reported gap.
\end{document}